\documentclass[11pt]{article}

\usepackage[T1]{fontenc}
\usepackage[utf8]{inputenc}
\usepackage{lmodern}
\usepackage{amsmath,amssymb,amsthm,mathtools}
\usepackage{booktabs}
\usepackage{enumitem}
\usepackage[margin=1in]{geometry}
\usepackage[protrusion=true,expansion=false]{microtype}
\usepackage[numbers,sort&compress]{natbib}
\usepackage[colorlinks=true,allcolors=blue]{hyperref}
\usepackage[nameinlink,noabbrev]{cleveref}

\newtheorem{theorem}{Theorem}[section]
\newtheorem{proposition}[theorem]{Proposition}
\newtheorem{lemma}[theorem]{Lemma}
\newtheorem{corollary}[theorem]{Corollary}
\theoremstyle{remark}
\newtheorem{remark}[theorem]{Remark}

\newcommand{\R}{\mathbb{R}}
\newcommand{\Sn}{\mathbb{S}}
\newcommand{\cS}{\mathcal{S}}

\newcommand{\Span}{\operatorname{span}}

\hypersetup{
  pdftitle={Rank-Three Projections and Minimal Multiplicity Bipartitions of Path Complements},
  pdfauthor={Jintao Fei; Jiangying Luo},
  pdfkeywords={inverse eigenvalue problem of a graph; multiplicity bipartition; tight frame; scalable frame; faithful orthogonal representation}
}

\title{Rank-Three Projections and Minimal Multiplicity Bipartitions of Path Complements}
\author{Jintao Fei\\
\small JD.com\\
\small \href{mailto:fei-jintao@outlook.com}{\texttt{fei-jintao@outlook.com}}
\and
Jiangying Luo\\
\small Tsinghua University\\
\small \href{mailto:luo-jiangying@outlook.com}{\texttt{luo-jiangying@outlook.com}}}
\date{}

\begin{document}

\maketitle

\begin{abstract}
For a graph $G$ admitting a real symmetric realization with exactly two
distinct eigenvalues, $MB(G)$ is the minimum, over all such realizations,
of the smaller of the two eigenvalue multiplicities.  Adm,
Fallat, Meagher, Nasserasr, Plosker, and Yang asked for this parameter for
the complement of a path on at least eight vertices.  We answer their
question completely by proving
\[
        MB(\overline{P_n})=3 \qquad (n\ge 6).
\]
In particular, this resolves the previously unresolved orders $n\ge9$
divisible by three.  The proof is exact and constructive.  We exhibit six
vectors in $\R^3$ whose mutual inner products vanish exactly for consecutive
indices, whose rank-one outer products form a basis of $\Sn^3$, and which
admit a strictly positive Parseval scaling.
An elementary absorption lemma then permits any finite faithful orthogonal
extension of this vector chain to be added with small positive weights while
the six original weights are corrected to retain the Parseval identity.
The resulting Gram matrix is a rank-three orthogonal projection in
$\cS(\overline{P_n})$.  A local two-dimensional orthogonality obstruction
gives the matching lower bound.  For completeness, we include self-contained
proofs of the exceptional small orders: $MB(\overline{P_3})=1$, whereas
$q(\overline{P_4})=4$ and $q(\overline{P_5})=3$.
\end{abstract}

\medskip
\noindent\textbf{Keywords.}
inverse eigenvalue problem of a graph; minimum number of distinct
eigenvalues; multiplicity bipartition; tight frame; scalable frame;
faithful orthogonal representation

\medskip
\noindent\textbf{2020 Mathematics Subject Classification.}
05C50, 15A18, 42C15.

\section{Introduction}

Let $G$ be a simple graph on $n$ vertices.  Denote by $\cS(G)$ the set of
real symmetric $n\times n$ matrices $A=[a_{ij}]$ such that, for $i\ne j$,
\[
 a_{ij}\ne0 \quad\Longleftrightarrow\quad \{i,j\}\in E(G),
\]
with no restriction on the diagonal.  The inverse eigenvalue problem of a
graph asks which spectra occur among matrices in $\cS(G)$.  One basic
parameter is
\[
 q(G)=\min\{q(A):A\in\cS(G)\},
\]
where $q(A)$ is the number of distinct eigenvalues of $A$; see, for example,
\cite{AhmadiEtAl2013,AdmEtAl2019,HogbenLinShader2022}.

Suppose $q(G)=2$.  If $G$ admits a matrix whose two eigenvalues have
multiplicities $n-k$ and $k$, with $1\le k\le \lfloor n/2\rfloor$, then
$[n-k,k]$ is an achievable multiplicity bipartition.  Following
\cite{AdmEtAl2019}, define
\[
 MB(G)=\min\bigl\{k:[n-k,k]\text{ is achievable by }G\bigr\}.
\]
The parameter has an equivalent frame-theoretic interpretation: $MB(G)$ is
the least rank, not exceeding $n/2$, of an orthogonal projection in
$\cS(G)$.  Thus its determination requires more than knowing merely that
$q(G)=2$ or knowing the minimum positive semidefinite rank.

Levene, Oblak, and \v{S}migoc proved that
$q(\overline{P_n})=2$ for $n\ge6$
\cite[Theorem~3.2]{LeveneOblakSmigoc2019}.  Adm et al.
subsequently showed
$MB(\overline{P_6})=MB(\overline{P_7})=3$ and posed the following question:
what is $MB(\overline{P_n})$ for $n\ge8$?
\cite[Question~5.1]{AdmEtAl2019}  The subsequent corrigendum
\cite{AdmEtAl2020Corrigendum} does not alter that question.  More recently,
in the range $n\ge6$ with $3\nmid n$, Barrett et al. recorded a rank-three
positive-semidefinite two-eigenvalue realization having the strong spectral property (SSP)
in the discussion following Porism~5.4
\cite[p.~160]{BarrettEtAl2026}.  Their discussion does not cover the residue
class $3\mid n$.  The main result of this paper gives a uniform,
self-contained answer for every $n\ge6$.

\begin{theorem}\label{thm:main}
For every integer $n\ge6$,
\[
             MB(\overline{P_n})=3.
\]
Equivalently, $\overline{P_n}$ has a realization that is a rank-three
orthogonal projection, and it has no realization by a projection of rank
one or two.
\end{theorem}

In particular, taking $n=3m$ yields
\[
             MB(\overline{P_{3m}})=3 \qquad (m\ge3).
\]

Our construction combines two ideas.  The first is a faithful orthogonal
representation in $\R^3$: consecutive path indices are represented by
orthogonal vectors and every other pair has nonzero inner product.  Finite
avoidance arguments of this kind also appear in
\cite[Lemma~5.2]{BarrettEtAl2026}.  Faithfulness alone, however, yields only
a rank-three positive semidefinite matrix; it does not force its three
positive eigenvalues to coincide.  The second ingredient addresses exactly
this gap.  Six explicit initial directions have rank-one outer products
spanning all of $\Sn^3$, while a strictly positive combination of those
outer products is $I_3$.  Any finite tail can therefore be added with a
small positive weight and absorbed by a small correction to the six anchor
weights.  The corrected vectors form a Parseval frame, so their Gram matrix
is an orthogonal projection.  This anchored absorption step is what permits
the faithful representation to remain on the same vertex set and have
exactly two eigenvalues.

The argument uses no floating-point calculation.  The seed vectors are
integral, all displayed identities are rational, and every extension
parameter may be chosen to be an integer.

\section{Projections, frames, and an absorption lemma}

We first record the standard projection formulation; compare
\cite[Lemma~2.2]{AdmEtAl2019} and the tight-frame formulation in
\cite{AbdollahiNajafi2018,ChenEtAl2014,FurstGrotts2021}.

\begin{proposition}[Projection--frame equivalence]\label{prop:projection}
Let $G$ be a graph on $n$ vertices and let
$1\le r\le\lfloor n/2\rfloor$.  The bipartition $[n-r,r]$ is achievable by
$G$ if and only if there is an $n\times r$ real matrix $U$ such that
\[
 U^{\mathsf T}U=I_r
 \quad\text{and}\quad
 UU^{\mathsf T}\in\cS(G).
\]
In this case $UU^{\mathsf T}$ is a rank-$r$ orthogonal projection with
spectrum $\{0^{(n-r)},1^{(r)}\}$.
\end{proposition}

\begin{proof}
If $A\in\cS(G)$ has two distinct eigenvalues $\lambda$ and $\mu$, with
multiplicities $n-r$ and $r$, respectively, then
\[
        P=\frac{A-\lambda I}{\mu-\lambda}
\]
is an orthogonal projection of rank $r$ and has the same off-diagonal zero
pattern as $A$.  Write $P=UU^{\mathsf T}$ using an orthonormal basis for its
range.  The converse follows immediately from the spectrum of
$UU^{\mathsf T}$ when $U^{\mathsf T}U=I_r$.
\end{proof}

Write $\Sn^d$ for the $d(d+1)/2$-dimensional real vector space of symmetric
$d\times d$ matrices.  The following elementary lemma isolates the weight
correction used later.  It is a strict-scalability statement in the language
of finite frames; see
\cite{KutyniokEtAl2013,CahillChen2013,AyyanarJohnsonThilak2025} for general
background on scalable frames and their associated graphs.

\begin{lemma}[Anchor absorption]\label{lem:absorption}
Let $x_1,\ldots,x_r\in\R^d$ be nonzero vectors such that
\[
 \Span\{x_i x_i^{\mathsf T}:1\le i\le r\}=\Sn^d
 \tag{2.1}\label{eq:span}
\]
and suppose that, for some $a_1,\ldots,a_r>0$,
\[
       \sum_{i=1}^r a_i x_i x_i^{\mathsf T}=I_d.
 \tag{2.2}\label{eq:anchorparseval}
\]
Then for every finite collection $y_1,\ldots,y_s\in\R^d$ there are strictly
positive weights $\alpha_1,\ldots,\alpha_r$ and
$\beta_1,\ldots,\beta_s$ satisfying
\[
 \sum_{i=1}^r \alpha_i x_i x_i^{\mathsf T}
 +\sum_{j=1}^s \beta_j y_j y_j^{\mathsf T}=I_d.
 \tag{2.3}\label{eq:absorbed}
\]
\end{lemma}

\begin{proof}
Set $T=\sum_{j=1}^s y_jy_j^{\mathsf T}$.  By \eqref{eq:span}, choose real
numbers $c_1,\ldots,c_r$ such that
\[
              -T=\sum_{i=1}^r c_i x_i x_i^{\mathsf T}.
\]
For $\varepsilon>0$, put
\[
 \alpha_i=a_i+\varepsilon c_i \quad(1\le i\le r),
 \qquad
 \beta_j=\varepsilon \quad(1\le j\le s).
\]
Since every $a_i$ is strictly positive, all $\alpha_i$ remain positive for
sufficiently small $\varepsilon>0$; the $\beta_j$ are positive as well.
Combining the two displayed identities gives \eqref{eq:absorbed}.
\end{proof}

\begin{remark}\label{rem:absorption}
Spanning, rather than linear independence, is the essential hypothesis in
\Cref{lem:absorption}.  If the $r=d(d+1)/2$ anchor outer products form a
basis, the correction coefficients are unique.  Strict positivity in
\eqref{eq:anchorparseval} supplies the open-neighborhood stability needed to
retain every anchor.  The lemma concerns each fixed finite tail; it neither
asserts a lower bound on the admissible weights uniform in $s$ nor implies
the strong spectral property.
\end{remark}

\section{Faithful path chains in three dimensions}

Call nonzero vectors $d_1,\ldots,d_k\in\R^3$ a \emph{faithful path chain}
if
\[
 d_i^{\mathsf T}d_j=0
 \quad\Longleftrightarrow\quad
 |i-j|=1
 \qquad(i\ne j).
 \tag{3.1}\label{eq:pathchain}
\]
Thus their Gram matrix has the off-diagonal zero pattern of $P_k$, or,
equivalently, its nonzero pattern represents $\overline{P_k}$.

\begin{lemma}[Finite-avoidance extension]\label{lem:extension}
Suppose $k\ge3$ and $d_1,\ldots,d_k\in\R^3$ form a faithful path chain, with
no two vectors parallel.  Then there is a vector $d_{k+1}$ such that
$d_1,\ldots,d_{k+1}$ retain both properties.  If the original vectors have
integer coordinates, $d_{k+1}$ can also be chosen with integer coordinates.
\end{lemma}

\begin{proof}
Put
\[
 b=d_{k-1},\qquad c=d_{k-1}\times d_k,
 \qquad v(t)=b+tc.
\]
Since $b$ and $d_k$ are nonzero and orthogonal, $c\ne0$; moreover $b\perp c$.
Consequently $v(t)\ne0$ for every real $t$, and
\[
 v(t)^{\mathsf T}d_k=0,
 \qquad
 v(t)^{\mathsf T}d_{k-1}=\lVert d_{k-1}\rVert^2\ne0.
\tag{3.2}\label{eq:newlast}
\]

For $j\le k-3$, the affine function
\[
 d_j^{\mathsf T}v(t)=d_j^{\mathsf T}d_{k-1}
       +t\,d_j^{\mathsf T}c
\]
has a nonzero constant term by \eqref{eq:pathchain}, so it has at most one
zero.  For $j=k-2$, its constant term is zero, but
\[
 d_{k-2}^{\mathsf T}c
 =d_{k-2}^{\mathsf T}(d_{k-1}\times d_k)\ne0.
\tag{3.3}\label{eq:triple}
\]
Indeed, equality in \eqref{eq:triple} would place $d_{k-2}$ in
$\Span(d_{k-1},d_k)$.  Both $d_{k-2}$ and $d_k$ are orthogonal to
$d_{k-1}$, so this would force $d_{k-2}\in\Span(d_k)$, contrary to the
nonparallel hypothesis.  Hence only finitely many values of $t$ create an
unwanted zero inner product.

It remains to avoid parallelism.  For each old vector $d_j$, at most one
value of $t$ can satisfy $v(t)\parallel d_j$.  Otherwise two distinct such
values, upon subtraction, would imply $c\parallel d_j$, and substitution
would then imply $b\parallel d_j$, impossible because $b$ and $c$ are
nonzero and orthogonal.  There are therefore only finitely many forbidden
values of $t$.  Choose any real $t$ outside this set and set
$d_{k+1}=v(t)$.  If the existing vectors are integral, the forbidden set is
still finite and $t$ may be chosen to be an integer.
\end{proof}

\begin{corollary}\label{cor:extend}
Every integral faithful path chain $d_1,\ldots,d_k$ in $\R^3$ whose vectors
are pairwise nonparallel extends, for every $N\ge k$, to an integral
faithful path chain $d_1,\ldots,d_N$ with the same property.
\end{corollary}

\section{The six-vector anchor and the main theorem}

The central exact configuration is the following.

\begin{proposition}[A six-vector absorbing anchor]\label{prop:anchor}
Let
\[
\begin{aligned}
 d_1&=(2,-1,0),      & d_2&=(0,0,1),       & d_3&=(1,1,0),\\
 d_4&=(1,-1,-2),     & d_5&=(0,2,-1),      & d_6&=(2,1,2),
\end{aligned}
\tag{4.1}\label{eq:seedvectors}
\]
and
\[
 (a_1,\ldots,a_6)
 =\left(\frac18,\frac38,\frac14,\frac1{12},\frac18,\frac1{24}\right).
\tag{4.2}\label{eq:seedweights}
\]
Then:
\begin{enumerate}[label=\textup{(\roman*)}]
\item $d_1,\ldots,d_6$ form a faithful path chain and are pairwise
      nonparallel;
\item $\sum_{i=1}^6 a_i d_i d_i^{\mathsf T}=I_3$;
\item the six matrices $d_i d_i^{\mathsf T}$ form a basis of $\Sn^3$.
\end{enumerate}
\end{proposition}

\begin{proof}
The exact Gram matrix is
\[
[d_i^{\mathsf T}d_j]_{i,j=1}^6=
\begin{pmatrix}
5&0&1&3&-2&3\\
0&1&0&-2&-1&2\\
1&0&2&0&2&3\\
3&-2&0&6&0&-3\\
-2&-1&2&0&5&0\\
3&2&3&-3&0&9
\end{pmatrix}.
\tag{4.3}\label{eq:gramseed}
\]
Its off-diagonal zeros occur exactly at consecutive pairs.  To verify the
nonparallel condition without approximation, the values
$\lVert d_i\times d_j\rVert^2$ for $1\le i<j\le6$, in lexicographic order,
are
\[
 5,9,21,21,36,\;2,2,4,5,\;12,6,9,\;30,45,\;45,
\]
all positive.  Direct multiplication using \eqref{eq:seedweights} gives
\[
             \sum_{i=1}^6 a_i d_i d_i^{\mathsf T}=I_3.
\tag{4.4}\label{eq:seedidentity}
\]

Finally define
\[
 \nu(x,y,z)=(x^2,y^2,z^2,xy,xz,yz)^{\mathsf T}.
\]
These are the coordinates of $(x,y,z)(x,y,z)^{\mathsf T}$ in the basis
\[
 (E_{11},E_{22},E_{33},E_{12}+E_{21},E_{13}+E_{31},E_{23}+E_{32})
\]
of $\Sn^3$.  Exact elimination gives
\[
 \det[\nu(d_1)\ \nu(d_2)\ \cdots\ \nu(d_6)]=288\ne0.
\tag{4.5}\label{eq:det288}
\]
Hence the six outer products form a basis.
\end{proof}

We can now prove the principal result.

\begin{proof}[Proof of \Cref{thm:main}]
Fix $n\ge6$.  By \Cref{cor:extend}, the six vectors in
\eqref{eq:seedvectors} extend to integral vectors
$d_1,\ldots,d_n$ satisfying
\[
 d_i^{\mathsf T}d_j=0
 \quad\Longleftrightarrow\quad |i-j|=1
 \qquad(i\ne j).
\tag{4.6}\label{eq:fullchain}
\]
If $n=6$ no extension is needed.  Apply \Cref{lem:absorption} with the first
six vectors as anchors and $d_7,\ldots,d_n$ as the finite tail.  By
\Cref{prop:anchor}, there exist weights $w_1,\ldots,w_n>0$ such that
\[
                 \sum_{i=1}^n w_i d_i d_i^{\mathsf T}=I_3.
\tag{4.7}\label{eq:fullparseval}
\]
Let $U$ be the $n\times3$ matrix whose $i$th row is
$u_i^{\mathsf T}=\sqrt{w_i}\,d_i^{\mathsf T}$.  Then
$U^{\mathsf T}U=I_3$, and hence
\[
                  P=UU^{\mathsf T}
\]
is a rank-three orthogonal projection.  For $i\ne j$,
\[
 P_{ij}=\sqrt{w_iw_j}\,d_i^{\mathsf T}d_j.
\]
All weights are positive, so \eqref{eq:fullchain} says precisely that
$P_{ij}=0$ if and only if $|i-j|=1$.  The consecutive pairs are exactly the
nonedges of $\overline{P_n}$, and therefore
$P\in\cS(\overline{P_n})$.  By \Cref{prop:projection},
\[
                  MB(\overline{P_n})\le3.
\]
This construction also proves directly that $q(\overline{P_n})=2$: the
projection has two eigenvalues, while $\overline{P_n}$ has an edge and hence
cannot be represented by a scalar matrix.

For the reverse inequality, suppose that
$P=UU^{\mathsf T}\in\cS(\overline{P_n})$ is a projection of rank
$r\in\{1,2\}$, and let $u_i^{\mathsf T}$ be row $i$ of $U$.  Every $u_i$ is
nonzero, since a zero row would make vertex $i$ isolated in the graph of
$P$, whereas $\overline{P_n}$ has no isolated vertex for $n\ge4$.  Moreover,
\[
                 u_i\perp u_{i+1}\qquad(1\le i<n).
\tag{4.8}\label{eq:roworth}
\]
Rank one is impossible because two nonzero vectors in $\R$ cannot be
orthogonal.  If $r=2$, then $u_1\perp u_2$ and $u_3\perp u_2$ force
$u_1\parallel u_3$.  Since $u_3\perp u_4$, it follows that
$u_1\perp u_4$.  But $\{1,4\}$ is an edge of $\overline{P_n}$, so
$P_{14}$ must be nonzero, a contradiction.  Thus no rank-one or rank-two
projection has the required zero pattern, and
$MB(\overline{P_n})\ge3$.
\end{proof}

\begin{remark}[Exactness and effectivity]\label{rem:exact}
At every extension step, \Cref{lem:extension} permits an integer value of
$t$, so all $d_i$ may be integral.  In \Cref{lem:absorption}, the coordinate
matrix of the six anchor outer products has determinant $288$; hence the
correction coefficients are rational.  The parameter $\varepsilon$ may be
chosen rational in a nonempty interval.  The only square roots enter when
passing from positive weights $w_i$ to the Parseval vectors
$\sqrt{w_i}d_i$.  Thus the proof is symbolic and does not rely on a finite
numerical search or on floating-point positivity tests.
\end{remark}

\section{Small orders and the complete classification}

For completeness we determine the orders below the range of
\Cref{thm:main}.  The values of $q(\overline{P_4})$ and
$q(\overline{P_5})$ also follow from
\cite[Theorem~3.2]{LeveneOblakSmigoc2019}; we give short self-contained
proofs and include $MB(\overline{P_3})$ to state the full classification.
We use the standard convention that $MB(G)$ is defined only when $q(G)=2$.

\begin{proposition}\label{prop:small}
For paths of orders $3,4,5$,
\[
 q(\overline{P_3})=2,\qquad MB(\overline{P_3})=1;
\]
\[
 q(\overline{P_4})=4,\qquad q(\overline{P_5})=3.
\]
Consequently $MB(\overline{P_4})$ and $MB(\overline{P_5})$ are undefined.
\end{proposition}

\begin{proof}
With the natural path labeling, $\overline{P_3}=K_2\cup K_1$, and
\[
 \frac12
 \begin{pmatrix}
 1&0&1\\0&0&0\\1&0&1
 \end{pmatrix}
 \in\cS(\overline{P_3})
\]
is a rank-one orthogonal projection.  It gives
$q(\overline{P_3})\le2$, while
$q(\overline{P_3})\ne1$ because the graph has an edge and hence no scalar
matrix belongs to $\cS(\overline{P_3})$.  Hence
$q(\overline{P_3})=2$ and $MB(\overline{P_3})=1$.

Next, $\overline{P_4}\cong P_4$, using the vertex order $3,1,4,2$.
After this permutation, every matrix in $\cS(P_4)$ is an irreducible real
symmetric tridiagonal matrix.  Such a matrix has simple spectrum: for fixed
$\lambda$, the first coordinate of a $\lambda$-eigenvector recursively
determines all remaining coordinates through the nonzero subdiagonal
entries.  Thus every eigenspace has dimension at most one, and symmetry
implies algebraic simplicity.  Hence $q(\overline{P_4})=4$.

For $\overline{P_5}$, suppose instead that $q(\overline{P_5})=2$.
Affine normalization at the eigenvalue of smaller multiplicity would, by
\Cref{prop:projection}, produce a projection in
$\cS(\overline{P_5})$ of rank at most $\lfloor5/2\rfloor=2$.  The rank-one
and rank-two obstructions from the proof of \Cref{thm:main} rule this out.
Thus $q(\overline{P_5})\ge3$.  Conversely, the exact matrix
\[
 A=\frac14
 \begin{pmatrix}
 0&0&5&4&3\\
 0&0&0&2&4\\
 5&0&0&0&5\\
 4&2&0&0&0\\
 3&4&5&0&0
 \end{pmatrix}
 \in\cS(\overline{P_5})
\tag{5.1}\label{eq:p5matrix}
\]
has characteristic polynomial
\[
 \chi_A(t)=\left(t-\frac52\right)
            \left(t^2+\frac54t-\frac58\right)^2.
\]
Its three distinct eigenvalues are
$5/2$ and $(-5\pm\sqrt{65})/8$.  Therefore
$q(\overline{P_5})\le3$, completing the proof.
\end{proof}

Combining \Cref{thm:main,prop:small} gives the complete classification for
$n\ge3$:
\[
MB(\overline{P_n})=
\begin{cases}
1, & n=3,\\
\text{undefined}, & n=4,5,\\
3, & n\ge6.
\end{cases}
\tag{5.2}\label{eq:classification}
\]
For $n=1,2$, the complement of the path is edgeless and has $q=1$.

\section{Relation to previous constructions}

It is useful to separate three increasingly restrictive tasks.  A faithful
orthogonal representation of $\overline{P_n}$ in $\R^3$ produces a
rank-three positive semidefinite matrix in $\cS(\overline{P_n})$.  A
two-eigenvalue realization additionally requires the frame operator to be a
scalar multiple of $I_3$.  Requiring the strong spectral property is a
further condition that is not needed for $MB$.

The finite-avoidance component of \Cref{lem:extension} is close in spirit to
the vector construction used for cycle complements in
\cite[Lemma~5.2]{BarrettEtAl2026}.  Their subsequent argument completes a
rank-three representation by adding three vertices.  Here the vertex set is
fixed.  The explicit seed in \Cref{prop:anchor} places $I_3$ in the relative
interior of the positive cone generated by six rank-one matrices that span
$\Sn^3$; \Cref{lem:absorption} then converts every finite faithful path
extension into a strictly scalable frame without adding vertices.  This is
the feature that removes the divisibility restriction and yields the exact
multiplicity bipartition.

Lin, Oblak, and \v{S}migoc proved the graph-level dichotomy
\[
 \overline{P_n}\in\mathcal G^{\mathrm{SSP}}
 \quad\Longleftrightarrow\quad 3\nmid n
\]
\cite[Theorem~5.5]{LinOblakSmigoc2020}.  Membership in
$\mathcal G^{\mathrm{SSP}}$ means that every matrix in
$\cS(\overline{P_n})$ has the strong spectral property.  Consequently,
nonmembership for $3\mid n$ is not an obstruction to the existence of a
rank-three projection; it says only that some matrix with this graph has no
SSP.  This distinction is one reason a direct projection construction is
needed.

In the discussion following Porism~5.4, Barrett et al. state that
$\overline{P_n}$ with $n>4$ and $3\nmid n$ is an example possessing the
required rank-three positive-semidefinite two-eigenvalue realization having the SSP
\cite[p.~160]{BarrettEtAl2026}.  The printed range includes $n=5$, which is
exceptional because $q(\overline{P_5})=3$ by \Cref{prop:small}; we therefore
use their statement only in the range $n\ge6$.  Our proof is independent of
that statement and applies simultaneously to all congruence classes.  In
particular, it supplies the orders $3\mid n$, $n\ge9$, left open by
\cite[Question~5.1]{AdmEtAl2019} and not covered by the Barrett et al.
example, thereby closing that question.

We have not established that the projections constructed here have the
SSP, and no SSP claim is needed for \Cref{thm:main}.  Determining whether
the anchor can be chosen so that the resulting projections enjoy an
appropriate strong property is a natural separate question.

\section{Conclusion}

The complement of every path on at least six vertices admits a rank-three
Parseval-frame realization, while its local pattern of three consecutive
path nonedges excludes ranks one and two.  Therefore its minimal
achievable multiplicity bipartition is $[n-3,3]$.  Together with the explicit
small-order analysis, this completely determines $MB(\overline{P_n})$
whenever the parameter is defined.  More broadly, the proof suggests a
reusable strategy: locate a faithful orthogonal representation containing a
strictly scalable set of rank-one anchors spanning the ambient symmetric
matrix space, and then absorb the remaining vertices by a small positive
perturbation of the anchor weights.

\bibliographystyle{plainnat}
\bibliography{references}

\begin{thebibliography}{13}
\providecommand{\natexlab}[1]{#1}
\providecommand{\url}[1]{\texttt{#1}}
\expandafter\ifx\csname urlstyle\endcsname\relax
  \providecommand{\doi}[1]{doi: #1}\else
  \providecommand{\doi}{doi: \begingroup \urlstyle{rm}\Url}\fi

\bibitem[Abdollahi and Najafi(2018)]{AbdollahiNajafi2018}
Farshid Abdollahi and Hashem Najafi.
\newblock Frame graph.
\newblock \emph{Linear and Multilinear Algebra}, 66\penalty0 (6):\penalty0
  1229--1243, 2018.
\newblock \doi{10.1080/03081087.2017.1347135}.

\bibitem[Adm et~al.(2019)Adm, Fallat, Meagher, Nasserasr, Plosker, and
  Yang]{AdmEtAl2019}
Mohammad Adm, Shaun Fallat, Karen Meagher, Shahla Nasserasr, Sarah Plosker, and
  Boting Yang.
\newblock Achievable multiplicity partitions in the inverse eigenvalue problem
  of a graph.
\newblock \emph{Special Matrices}, 7:\penalty0 276--290, 2019.
\newblock \doi{10.1515/spma-2019-0022}.

\bibitem[Adm et~al.(2020)Adm, Fallat, Meagher, Nasserasr, Plosker, and
  Yang]{AdmEtAl2020Corrigendum}
Mohammad Adm, Shaun Fallat, Karen Meagher, Shahla Nasserasr, Sarah Plosker, and
  Boting Yang.
\newblock Corrigendum to ``{Achievable multiplicity partitions in the inverse
  eigenvalue problem of a graph}''.
\newblock \emph{Special Matrices}, 8\penalty0 (1):\penalty0 235--241, 2020.
\newblock \doi{10.1515/spma-2020-0117}.

\bibitem[Ahmadi et~al.(2013)Ahmadi, Alinaghipour, Cavers, Fallat, Meagher, and
  Nasserasr]{AhmadiEtAl2013}
Bahman Ahmadi, Fatemeh Alinaghipour, Michael~S. Cavers, Shaun Fallat, Karen
  Meagher, and Shahla Nasserasr.
\newblock Minimum number of distinct eigenvalues of graphs.
\newblock \emph{Electronic Journal of Linear Algebra}, 26:\penalty0 673--691,
  2013.
\newblock \doi{10.13001/1081-3810.1679}.

\bibitem[Ayyanar et~al.(2025)Ayyanar, Johnson, and
  Thilak]{AyyanarJohnsonThilak2025}
K.~Ayyanar, P.~Sam Johnson, and A.~Senthil Thilak.
\newblock Frame scaling by graphs.
\newblock \emph{Indian Journal of Pure and Applied Mathematics}, 56\penalty0
  (3):\penalty0 891--900, 2025.
\newblock \doi{10.1007/s13226-025-00809-2}.

\bibitem[Barrett et~al.(2026)Barrett, Fallat, Furst, Nasserasr, Rooney, and
  Tait]{BarrettEtAl2026}
Wayne Barrett, Shaun Fallat, Veronika Furst, Shahla Nasserasr, Brendan Rooney,
  and Michael Tait.
\newblock Graphs with bipartite complement that admit two distinct eigenvalues.
\newblock \emph{Electronic Journal of Linear Algebra}, 42:\penalty0 146--161,
  2026.
\newblock \doi{10.13001/ela.2026.9443}.

\bibitem[Cahill and Chen(2013)]{CahillChen2013}
Jameson Cahill and Xuemei Chen.
\newblock A note on scalable frames.
\newblock In \emph{Proceedings of the 10th International Conference on Sampling
  Theory and Applications}, pages 93--96, 2013.

\bibitem[Chen et~al.(2014)Chen, Grimm, McMichael, and Johnson]{ChenEtAl2014}
Zhao Chen, Matthew Grimm, Paul McMichael, and Charles~R. Johnson.
\newblock Undirected graphs of {H}ermitian matrices that admit only two
  distinct eigenvalues.
\newblock \emph{Linear Algebra and its Applications}, 458:\penalty0 403--428,
  2014.
\newblock \doi{10.1016/j.laa.2014.02.022}.

\bibitem[Furst and Grotts(2021)]{FurstGrotts2021}
Veronika Furst and Howard Grotts.
\newblock Tight frame graphs arising as line graphs.
\newblock \emph{The PUMP Journal of Undergraduate Research}, 4:\penalty0 1--19,
  2021.
\newblock \doi{10.46787/pump.v4i0.2415}.

\bibitem[Hogben et~al.(2022)Hogben, Lin, and Shader]{HogbenLinShader2022}
Leslie Hogben, Jephian C.-H. Lin, and Bryan~L. Shader.
\newblock \emph{Inverse Problems and Zero Forcing for Graphs}, volume 270 of
  \emph{Mathematical Surveys and Monographs}.
\newblock American Mathematical Society, Providence, RI, 2022.
\newblock \doi{10.1090/surv/270}.

\bibitem[Kutyniok et~al.(2013)Kutyniok, Okoudjou, Philipp, and
  Tuley]{KutyniokEtAl2013}
Gitta Kutyniok, Kasso~A. Okoudjou, Friedrich Philipp, and Elizabeth~K. Tuley.
\newblock Scalable frames.
\newblock \emph{Linear Algebra and its Applications}, 438\penalty0
  (5):\penalty0 2225--2238, 2013.
\newblock \doi{10.1016/j.laa.2012.10.046}.

\bibitem[Levene et~al.(2019)Levene, Oblak, and
  \v{S}migoc]{LeveneOblakSmigoc2019}
Rupert~H. Levene, Polona Oblak, and Helena \v{S}migoc.
\newblock A {N}ordhaus--{G}addum conjecture for the minimum number of distinct
  eigenvalues of a graph.
\newblock \emph{Linear Algebra and its Applications}, 564:\penalty0 236--263,
  2019.
\newblock \doi{10.1016/j.laa.2018.12.001}.

\bibitem[Lin et~al.(2020)Lin, Oblak, and \v{S}migoc]{LinOblakSmigoc2020}
Jephian C.-H. Lin, Polona Oblak, and Helena \v{S}migoc.
\newblock The strong spectral property for graphs.
\newblock \emph{Linear Algebra and its Applications}, 598:\penalty0 68--91,
  2020.
\newblock \doi{10.1016/j.laa.2020.03.031}.

\end{thebibliography}

\end{document}